\documentclass[letterpaper,10pt,conference]{ieeeconf}  
\IEEEoverridecommandlockouts                              
\usepackage{cite}
\usepackage{amsmath,amssymb,amsfonts}
\usepackage{graphicx}
\usepackage{xcolor}
\let\labelindent\relax
\usepackage{enumitem}
\usepackage[normalem]{ulem}
\makeatletter
\let\NAT@parse\undefined
\makeatother
\usepackage[
    colorlinks=true,
    linkcolor=blue,
    citecolor=blue,
    urlcolor=blue
]{hyperref}

\graphicspath{{figures/}}

\definecolor{myred}{rgb}{0.7,0.1,0.16}

\newtheorem{theorem}{Theorem}
\newtheorem{proposition}{Proposition}
\newtheorem{corollary}{Corollary}
\newtheorem{lemma}{Lemma}
\newtheorem{definition}{Definition}

\newtheorem{example}{Example}

\newcommand{\R}{\mathbb{R}}
\newcommand{\C}{\mathcal{C}}

\newcommand{\Lie}[2]{L_{#1}#2}

\catcode`\@=11
\def\downparenfill{$\m@th\braceld\leaders\vrule\hfill\bracerd$}
\def\overparen#1{\mathop{\vbox{\ialign{##\crcr\crcr \noalign{\kern0.4ex}
\downparenfill\crcr\noalign{\kern0.4ex\nointerlineskip}
$\hfil\displaystyle{#1}\hfil$\crcr}}}\limits}
\catcode`\@=12

\title{\LARGE \bf
Construction of Control Lyapunov–Barrier Functions\\ from CLF--CBF Pairs
}

\author{Bo Wang and Miroslav Krsti{\'c}
\thanks{Bo Wang is with the Department of Mechanical Engineering, The City College of New York, The City University of New York, New York, NY 10031, USA (e-mail: bwang1@ccny.cuny.edu). }
\thanks{Miroslav Krsti{\'c} is with the Department of Mechanical and Aerospace Engineering, University of California San Diego, La Jolla, CA 92093, USA (e-mail: mkrstic@ucsd.edu). }
}

\begin{document}

\maketitle
\thispagestyle{empty}
\pagestyle{empty}

\begin{abstract}
This paper studies the construction of control Lyapunov--barrier functions (CLBFs) from a given control Lyapunov function (CLF) $V$ and control barrier function (CBF) $h$. We consider functions of the form $W=F(V,h)$ that increase with the CLF value and do not increase with the barrier value, and show that the CLBF decrease condition is completely characterized by a nonnegative scalar weight that governs the relative contributions of the CLF and CBF. Because this weight depends only on $(V,h)$, the same value must satisfy the decrease condition at all states sharing the same CLF and CBF values, which leads to a joint-level-set admissibility condition. 
We then construct CLBFs from admissible weights through a first-order partial differential equation (PDE) and obtain explicit multiplicative and power-type families as special cases. We further identify an integrability obstruction showing that admissibility alone does not guarantee properness. Two nonlinear examples illustrate the constructions and demonstrate safe stabilization in cases where feedback based on the original CLF violates the safety constraint.
\end{abstract}

\section{Introduction}\label{sec:introduction}

Safe stabilization seeks to asymptotically stabilize a desired equilibrium while ensuring that the system trajectory remains within a prescribed safe set. Control Lyapunov functions (CLFs) and control barrier functions (CBFs) provide standard tools for these objectives, respectively characterizing stabilizability through a decrease condition and safety through a forward-invariance condition. A widely used approach enforces the CLF and CBF inequalities simultaneously, most notably through CLF--CBF quadratic programs (QPs) \cite{ames2017control,ames2019control}. Treating the CLF and CBF conditions separately provides design flexibility, but the two inequalities may conflict, and the existence of a control input satisfying both is not guaranteed. This has motivated the study of CLF--CBF compatibility \cite{dai2024verification}.

Another line of work constructs universal feedback laws that enforce the CLF and CBF inequalities directly. Existing approaches construct smooth safe stabilizing feedback from weighted centroids of feasible control sets~\cite{ong2019universal}, characterize pointwise CLF--CBF compatibility for single-input systems and derive closed-form universal feedback laws~\cite{wang2026universal}, 
and obtain smooth feedback satisfying an arbitrary number of strictly feasible affine inequalities through strictly convex minimization, together with neural-network approximations~\cite{mestres2026universal}. These methods address simultaneous constraint satisfaction and feedback regularity.

An alternative is to encode stability and safety within a single \textit{control Lyapunov--barrier function} (CLBF), thereby replacing two potentially competing inequalities with one decrease condition. An early CLBF construction for safe stabilization combined a separately designed CLF and CBF~\cite{romdlony2016stabilization}. Subsequent analysis identified restrictive existence conditions and showed that some assumptions underlying that construction cannot hold as originally stated~\cite{braun2020comment}. More recent work has developed converse results for stability and safety within a unified framework. Smooth converse Lyapunov--barrier theorems establish the existence of a single smooth Lyapunov--barrier function under appropriate stability-with-safety properties~\cite{meng2022smooth}, while related results connect safe stabilizability with the existence of CLBFs and compatible CLF--CBF pairs~\cite{mestres2025converse}. Recent developments further relate strictly compatible CLF--CBF pairs to single smooth Lyapunov--barrier functions~\cite{quartz2026converse}, provide explicit smooth-patching constructions for combining a given CLF and CBF under strict compatibility conditions~\cite{liu2025computing}, and characterize maximal CLBFs through Zubov equations under compatibility and additional regularity conditions~\cite{meng2026zubov}.

In this paper, we consider the construction of CLBFs of the form $W=F(V,h)$ from a given CLF $V$ and CBF $h$. Inverse optimal control provides an additional motivation for this form. Existing inverse-optimal designs use $L_gV$ for stabilization~\cite{krstic1998inverse} and $L_gh$ for safety filtering~\cite{krstic2024inverse}. A CLBF $W=F(V,h)$ instead provides a single $L_gW$ term, suggesting a natural route toward \textit{inverse-optimal safe stabilization}.
We address three questions: when does such a function satisfy the CLBF decrease condition, what can prevent it from being proper on the interior of the safe set, and how can $F$ be constructed? At the core of the analysis is a nonnegative scalar weight $\kappa(v,s):=-{F_s(v,s)}/{F_v(v,s)}$ for functions $F(V,h)$ that increase with the CLF value and do not increase
with the barrier value. This weight completely determines the CLBF decrease condition. Because $\kappa$ depends only on $(V,h)$, the same value must satisfy the decrease condition at every state sharing the same CLF and CBF values, yielding a joint-level-set characterization of admissible weights.

The main results are constructive. Given an admissible weight $\kappa$, the function $F$ is constructed by solving a first-order partial differential equation (PDE). This gives a general construction and yields explicit multiplicative families, as well as a simple power-type family under CLF--CBF compatibility and additional boundary-growth conditions. We further show that weight admissibility alone does not guarantee properness, and identify an integrability-based obstruction to properness for certain admissible weights. Two nonlinear examples illustrate the constructions, including a nonmultiplicative one, and demonstrate safe stabilization in cases where feedback based on the original CLF violates the safety constraint.

\section{Preliminaries and Problem Formulation}
\label{sec:preliminaries}

\textit{Notation}: Let $|\cdot|$ denote the Euclidean norm. For $S\subset\mathbb{R}^n$, $\partial S$ and $\operatorname{Int}(S)$ denote its boundary and interior, respectively. The class $\mathcal{K}$ consists of all continuous, strictly increasing functions $\alpha:\mathbb{R}_{\geq 0}\to\mathbb{R}_{\geq 0}$ satisfying $\alpha(0)=0$. The extended class $\mathcal{K}$, denoted by $\mathcal{K}^e$, consists of all continuous, strictly increasing functions $\alpha_h:\mathbb{R}\to\mathbb{R}$ satisfying $\alpha_h(0)=0$.

Consider the control-affine system
\begin{equation}
    \dot x=f(x)+g(x)u,
    \label{eq:system}
\end{equation}
where $x\in\R^n$ is the state, $u\in\R^m$ is the control input, and $f:\R^n\to\R^n$ and $g:\R^n\to\R^{n\times m}$ are locally Lipschitz. We assume $f(0) = 0$ so that the origin is an equilibrium of the unforced system. Let $h:\R^n\to\R$ be continuously differentiable and define
\begin{equation}
    \C:=\{x\in\R^n:h(x)\ge 0\}.
    \label{eq:safe-set}
\end{equation}
Assume that $\operatorname{Int}(\C):=\{x\in\R^n:h(x)>0\}$ is connected and contains the origin. Let $\mathcal D\subseteq\mathbb R^n$ be an open set such that $\C\subseteq\mathcal D$.

\begin{definition}[CLF]\rm 
A continuously differentiable, positive definite, and proper function $V:\mathcal{D}\to \mathbb{R}_{\ge 0}$ is a \textit{CLF} for \eqref{eq:system} on $\mathcal{D}$ if there exists a continuous and positive definite function $\alpha:\mathbb{R}_{\ge 0}\to \mathbb{R}_{\ge 0}$ such that for all $x\in\mathcal{D}\backslash \{0\}$
\begin{equation}\label{eq:defCLF}
    L_gV(x)=0 \implies L_fV(x) + \alpha(|x|)<0.
\end{equation}
If $\mathcal{D}=\mathbb{R}^n$, then $V$ is a global CLF for \eqref{eq:system}.
\end{definition}

\begin{definition}[CBF] \rm
A continuously differentiable function $h:\mathbb{R}^n\to \mathbb{R}$ is a \textit{CBF} with respect to the set $\mathcal{C}$ if there exists $\alpha_h\in\mathcal{K}^e$ such that
\begin{equation}\label{eq:defCBF}
    L_gh(x)=0 \implies L_fh(x) + \alpha_h(h(x))>0
\end{equation}
for all $x\in\mathbb{R}^n$.
\end{definition}

Different definitions of CLBFs have been proposed in the literature; see, e.g., \cite{romdlony2016stabilization,quartz2026converse}. Here, we use a CLBF defined on the open safe domain $\operatorname{Int}{(\mathcal{C})}$, with properness accounting for both approach to the safety boundary and escape to infinity, and with a decrease condition providing feedback stabilization \cite{tee2009barrier,meng2022smooth}.

\begin{definition}[CLBF]\label{def:CLBF}
A continuously differentiable function $W:\operatorname{Int}(\C)\to\R_{\geq 0}$ is a \textit{CLBF} for \eqref{eq:system} on $\operatorname{Int}(\C)$ if:
\begin{enumerate}
    \item $W(0)=0$ and $W(x)>0$ for every $x\in\operatorname{Int}(\C)\setminus\{0\}$;

    \item $W$ is proper on $\operatorname{Int}(\C)$; that is, for every $c\geq0$, the sublevel set $\Omega_c(W):=\{x\in\operatorname{Int}(\C):W(x)\leq c\}$ is a compact subset of $\operatorname{Int}(\C)$;

    \item for every $x\in\operatorname{Int}(\C)\setminus\{0\}$,
    \begin{equation}
        \Lie{g}{W}(x)=0
        \quad\implies\quad
        \Lie{f}{W}(x)<0.
        \label{eq:clbf-artstein}
    \end{equation}
\end{enumerate}
\end{definition}

Properness is equivalent to $W(x_i)\to\infty$ along every sequence in $\operatorname{Int}(\C)$ leaving all compact subsets of the domain. It therefore captures both approach to $\partial\C$ and, when $\operatorname{Int}(\C)$ is unbounded, escape to infinity. We use the strict decrease condition in \eqref{eq:clbf-artstein} without imposing a prescribed quantitative decay margin. For unconstrained input $u\in\R^m$, \eqref{eq:clbf-artstein} is the standard control-Lyapunov decrease condition for feedback stabilization \cite{sontag1989universal}.  
Consequently, any closed-loop trajectory initialized in $\operatorname{Int}(\C)$ along which $W$ is nonincreasing remains in the compact sublevel set and therefore remains safe.

\textit{Problem Statement}: Given a CLF $V$ and a CBF $h$, we seek to construct a continuously differentiable function $F$ of two variables, defined on an open set containing the values attained by $(V,h)$ on $\operatorname{Int}(\mathcal C)$, such that
\begin{equation}
    W(x)=F(V(x),h(x))
    \label{eq:composition}
\end{equation}
is a CLBF on $\operatorname{Int}(\mathcal C)$. To do so, we address three questions: 
1) when can a function of this form satisfy the CLBF decrease condition,
2) what can obstruct properness on the interior of the safe set, and
3) how can $F$ be constructed?

In the next section, we consider functions $F(V,h)$ that increase with the CLF value and do not increase with the barrier value. We show that their decrease condition is determined by a single nonnegative scalar weight.

\section{Scalar Weight Characterization}\label{sec:characterization}

\subsection{The Decrease Condition}

We first examine how the partial derivatives of $F$ enter the CLBF decrease condition \eqref{eq:clbf-artstein}. Define the joint value set of $(V,h)$ on $\operatorname{Int}(\C)$ as
\begin{equation}
    \mathcal{R}_{V,h}:=\left\{(V(x),h(x)):x\in\operatorname{Int}(\C)\right\}.
    \label{eq:joint-value-set}
\end{equation}
Consider a continuously differentiable function $F(v,s)$, defined on an open set containing $\mathcal{R}_{V,h}$, and suppose that
\begin{equation}
    F_v(v,s)>0, \qquad F_s(v,s)\leq0
    \label{eq:F-monotonicity}
\end{equation}
for every $(v,s)\in\mathcal{R}_{V,h}$. Thus, $F$ increases with the CLF value and does not increase with the barrier value.

The chain rule gives
\begin{subequations}
\begin{align}
    \Lie{f}{W}
    &=F_v(V,h)\Lie{f}{V}+F_s(V,h)\Lie{f}{h},
    \label{eq:LfW-chain}\\
    \Lie{g}{W}
    &=F_v(V,h)\Lie{g}{V}+F_s(V,h)\Lie{g}{h}.
    \label{eq:LgW-chain}
\end{align}
\end{subequations}
Since $F_v>0$, only the ratio of the two partial derivatives affects the CLBF decrease condition. Define the nonnegative scalar \textit{weight}
\begin{equation}
    \kappa(v,s):=-\frac{F_s(v,s)}{F_v(v,s)}\ge 0.
    \label{eq:weight}
\end{equation}
Then
\begin{subequations}
\begin{align}
    \Lie{f}{W}
    &=F_v(V,h)
    \left[\Lie{f}{V}-\kappa(V,h)\Lie{f}{h}\right],
    \label{eq:LfW-weight}\\
    \Lie{g}{W}
    &=F_v(V,h)
    \left[\Lie{g}{V}-\kappa(V,h)\Lie{g}{h}\right].
    \label{eq:LgW-weight}
\end{align}
\end{subequations}
Thus, $\kappa$ determines the relative weighting of the CBF and the CLF gradients in $\nabla W$, while $F_v$ provides only a positive scaling. In particular, $\kappa=0$ recovers the CLF decrease condition, whereas $\kappa>0$ introduces dependence on the barrier value.

The weight $\kappa(V(x),h(x))$ depends on $x$ only through $V$ and $h$. For each $(v,s)\in\mathcal{R}_{V,h}$, define the joint level set
\begin{equation}
    \Sigma(v,s):=\left\{x\in\operatorname{Int}(\C):V(x)=v,\ h(x)=s\right\}.
    \label{eq:joint-level-set}
\end{equation}
All states in $\Sigma(v,s)$ must use the same value $\kappa(v,s)$. Accordingly, define the admissible set
\begin{equation}
\mathcal A(v,s):=\left\{
k\geq 0:
\begin{array}{l}
L_gV(x)-kL_gh(x)=0\\[1mm]
\implies
L_fV(x)-kL_fh(x)<0,\\[1mm]
\forall x\in\Sigma(v,s)\setminus\{0\}
\end{array}
\right\}.
\label{eq:admissible-set}
\end{equation}
Thus, $\mathcal A(v,s)$ contains exactly the scalar weights that satisfy the decrease condition \eqref{eq:clbf-artstein} simultaneously at all states sharing the same pair $(v,s)$. 

\begin{lemma}[Scalar-Weight Characterization]
\label{lem:weight-characterization}
Suppose that \eqref{eq:F-monotonicity} holds and that $W=F(V,h)$ is positive definite and proper on $\operatorname{Int}(\C)$. Then $W$ is a CLBF if and only if {\setlength{\abovedisplayskip}{4pt}
\begin{equation}\label{eq:weight-selection}
    \kappa(v,s)\in\mathcal A(v,s), \quad \forall (v,s)\in \mathcal R_{V,h}.
\end{equation}}
\end{lemma}
\vspace{4pt}

\begin{proof}
By \eqref{eq:LfW-weight}--\eqref{eq:LgW-weight} and $F_v(v,s)>0$, the CLBF decrease condition \eqref{eq:clbf-artstein} is equivalent to $[L_gV-\kappa(V,h)L_gh=0] \implies [L_fV-\kappa(V,h)L_fh<0]$ at every nonzero state. For a fixed $(v,s)$, the same value $\kappa(v,s)$ must satisfy this implication for every $x\in\Sigma(v,s)\setminus\{0\}$. By the definition of $\mathcal A(v,s)$, this is exactly \eqref{eq:weight-selection}.
\end{proof}

Thus, once a function $\kappa(v,s)$ is chosen, its admissibility can be checked directly in the state space. As the following example shows, the admissible set $\mathcal A(v,s)$ need not be an interval.

\begin{example}[Admissible Weights for a Linear System]\label{ex:1}
    Consider the linear system
    \begin{equation}\label{eq:linear-example-system}
        \dot x=Ax+Bu
    \end{equation}
    with
\begin{equation*}
    A=\begin{bmatrix}
        2&1&1\\
        -1&-1&-1\\
        -1&1&-1
    \end{bmatrix}
    \quad\text{and}\quad
    B=\begin{bmatrix}
        1\\0\\0
    \end{bmatrix}.
\end{equation*}
The matrix $A$ has the positive eigenvalue $2^{1/3}$. Let
\begin{equation*}
    V(x)=\frac{1}{2}(x_1^2+x_2^2+x_3^2),
    \qquad
    h(x)=1-x_1-x_2.
\end{equation*}
The corresponding Lie derivatives are $\Lie{f}{V}=2x_1^2-x_2^2-x_3^2$, $\Lie{g}{V}=x_1$, $\Lie{f}{h}=-x_1$, and $\Lie{g}{h}=-1$. Hence, $\Lie{g}{V}=0$ implies $\Lie{f}{V}=-x_2^2-x_3^2<0$ away from the origin, so $V$ is a CLF. Since $\Lie{g}{h}=-1$, $h$ is a CBF for the half-space $\C=\{x:x_1+x_2\leq1\}$. The identities $h(x)=s$ and $V(x)=v$ give $x_1+x_2=1-s$ and $x_1^2+x_2^2+x_3^2=2v$. The minimum of $V$ subject to $x_1+x_2=1-s$ is $(1-s)^2/4$, attained at $x_1=x_2=(1-s)/2$ and $x_3=0$. Therefore,
\begin{equation}
    \mathcal R_{V,h}
    =\left\{(v,s):s>0,\quad v\geq\frac{(1-s)^2}{4}\right\}.
    \label{eq:linear-example-R}
\end{equation}
Fix $(v,s)\in\mathcal R_{V,h}$. If $\Lie{g}{V}-k\Lie{g}{h}=0$, then $x_1=-k$, $x_2=1-s+k$, and $x_3^2 =2v-k^2-(1-s+k)^2$. At such a state,
\begin{equation}
    \Lie{f}{V}-k\Lie{f}{h}=2(k^2-v).
    \label{eq:linear-example-drift}
\end{equation}
If $0<s<1$, then $x_3^2\ge 0$ implies $2v\geq k^2+(1-s+k)^2>2k^2$. Thus, whenever the equality $\Lie{g}{V}-k\Lie{g}{h}=0$ can occur, $k<\sqrt v$, and \eqref{eq:linear-example-drift} is negative. Therefore,
\begin{equation}
    \mathcal A(v,s)=\R_{\geq0}, \qquad 0<s<1.
    \label{eq:linear-example-A-boundary}
\end{equation}
For $s\geq1$ and $(v,s)\neq(0,1)$, define
\begin{equation}
    k_+(v,s):=\frac{s-1+\sqrt{4v-(s-1)^2}}{2}.
    \label{eq:linear-example-kplus}
\end{equation}
The condition $x_3^2\ge 0$ defines an interval of feasible values of $k$ with upper endpoint $k_+(v,s)$. Combining with \eqref{eq:linear-example-drift} gives
\begin{equation}
    \mathcal A(v,s)
    =[0,\sqrt v)\cup(k_+(v,s),\infty),  \qquad s\geq1,
    \label{eq:linear-example-A-interior}
\end{equation}
with $\mathcal A(0,1)=\R_{\geq0}$. Thus, the admissible set can be disconnected even for a linear system with a quadratic CLF and a linear CBF.

A smooth selection is
\begin{equation}
    \kappa(v,s):=\frac{v}{s(1+v)}.
    \label{eq:linear-example-kappa}
\end{equation}
For $0<s<1$, admissibility follows directly from \eqref{eq:linear-example-A-boundary}. For $s\geq1$ and $(v,s)\neq(0,1)$, using $2\sqrt v\leq1+v$ gives $\kappa(v,s) \leq{\sqrt v}/{(2s)} <\sqrt v$, so $\kappa(v,s)$ lies in the lower branch of \eqref{eq:linear-example-A-interior}. Hence $\kappa(v,s)\in\mathcal A(v,s)$ for all $(v,s)\in\mathcal R_{V,h}$.
\hfill $\blacktriangle$
\end{example}

\subsection{An Obstruction to Properness}
\label{sec:properness-obstruction}

Since $V$ is a CLF, $0\in\mathcal A(v,s)$ for every $(v,s)\in\mathcal R_{V,h}$, so the admissible sets are nonempty. The decrease condition and properness, however, impose different requirements on the weight $\kappa$. The choice $\kappa=0$ is always admissible but introduces no dependence on the barrier value. More generally, a particular admissible selection may still grow too slowly as $s\to0$ to produce a proper function.

\begin{proposition}[Integrability Obstruction to Properness]
\label{prop:weight-growth-obstruction}
Suppose that there exist $0<v_0<v_1$ and $s_0>0$ such that $[v_0,v_1]\times(0,s_0]\subseteq\mathcal R_{V,h}$. Let $\kappa:[v_0,v_1]\times (0,s_0]\to\mathbb{R}_{\ge 0}$ be continuous. Define
\begin{equation}
    \kappa_{\max}(s):=\max_{v\in[v_0,v_1]}\kappa(v,s).
    \label{eq:kappa-max-selection}
\end{equation}
If
\begin{equation}
    \int_0^{s_0}\kappa_{\max}(s)\,\mathrm{d}s<\infty,
    \label{eq:kappa-max-integrable}
\end{equation}
then any continuously differentiable function $F$ satisfying \eqref{eq:F-monotonicity} and $-F_s/F_v=\kappa$ throughout $[v_0,v_1]\times (0,s_0]$ fails to make $W=F(V,h)$ proper on $\operatorname{Int}(\C)$.
\end{proposition}

\begin{proof}
Since $\kappa_{\max}$ is nonnegative and integrable on $(0,s_0]$, there exists $\bar s\in(0,s_0]$ such that
\begin{equation}
    \eta:= \int_0^{\bar s}\kappa_{\max}(s)\,\mathrm{d}s <v_1-v_0.
    \label{eq:obstruction-eta}
\end{equation}
Fix $\sigma\in(0,\bar s]$, and consider
\begin{equation}
    \frac{\mathrm{d}\xi}{\mathrm{d}\tau} =  \kappa(\xi(\tau),\tau), \qquad \xi(\sigma)=v_0.
    \label{eq:obstruction-curve}
\end{equation}
By continuity of $\kappa$, a solution exists locally. As long as $\xi(\tau)\in[v_0,v_1]$, the definition of $\kappa_{\max}$ gives
\begin{equation}
    0\leq \frac{\mathrm{d}\xi}{\mathrm{d}\tau} \leq\kappa_{\max}(\tau),
\end{equation}
and hence
\begin{equation}
    v_0 \leq\xi(\tau) \leq v_0+\int_\sigma^\tau \kappa_{\max}(r)\,\mathrm{d}r \leq v_0+\eta <v_1.
    \label{eq:obstruction-curve-bound}
\end{equation}
Thus the solution cannot leave $[v_0,v_1]$ before reaching $\tau=\bar s$, and therefore it extends to the whole interval $[\sigma,\bar s]$.

Since $-F_s/F_v=\kappa$ and $F_v>0$, one has
\begin{equation}
    F_s(v,s)+\kappa(v,s)F_v(v,s)=0.
\end{equation}
throughout the considered region. Along \eqref{eq:obstruction-curve},
\begin{equation}
    \frac{\mathrm{d}}{\mathrm{d}\tau} F(\xi(\tau),\tau) = F_v(\xi(\tau),\tau) \frac{\mathrm{d}\xi}{\mathrm{d}\tau} +F_s(\xi(\tau),\tau) =0.
\end{equation}
Therefore,
\begin{equation}
    F(v_0,\sigma) = F(\xi(\bar s),\bar s).
    \label{eq:obstruction-constant}
\end{equation}
By \eqref{eq:obstruction-curve-bound}, $\xi(\bar s)\in[v_0,v_0+\eta]$. Since $F(\cdot,\bar s)$ is continuous, there exists a finite constant $c$ such that $F(v_0,\sigma)\leq c$ for all $\sigma\in(0,\bar s]$.

Now choose a sequence $\sigma_j\to0^+$. Since $(v_0,\sigma_j)\in\mathcal R_{V,h}$, there exists $x_j\in\operatorname{Int}(\C)$ such that $V(x_j)=v_0$ and $h(x_j)=\sigma_j$. Hence
\begin{equation}
    W(x_j)=F(v_0,\sigma_j)\leq c.
\end{equation}
If $W$ were proper, the points $x_j$ would all belong to a compact sublevel set of $W$ contained in $\operatorname{Int}(\C)$. Since $h>0$ on $\operatorname{Int}(\C)$ and $h$ is continuous, $h$ would have a positive minimum on this sublevel set. This contradicts $h(x_j)=\sigma_j\to0^+$. Therefore, $W$ cannot be proper.
\end{proof}

Proposition \ref{prop:weight-growth-obstruction} shows that if the weight $\kappa(v,s)$ is integrable in $s$ near the safety boundary, it cannot generate the boundary growth required for properness. The distinction can already be seen in the preceding linear example. An alternative admissible selection is
\begin{equation}
    \tilde{\kappa}(v,s):=\frac{v}{\sqrt{s}(1+v)}\in\mathcal{A}(v,s).
    \label{eq:linear-example-kappa-integrable}
\end{equation}
However, on any strip $[v_0,v_1]\times(0,s_0]\subseteq\mathcal R_{V,h}$, $\tilde{\kappa}_{\max}(s)$ is proportional to $1/\sqrt{s}$ and is therefore integrable at $s=0$. Proposition~\ref{prop:weight-growth-obstruction} shows that this selection cannot produce a proper CLBF. In contrast, the selection \eqref{eq:linear-example-kappa} has $1/s$-type growth and is not ruled out by the proposition.

\section{Construction of CLBFs}\label{sec:construction}

\subsection{Construction via Characteristics}
\label{sec:construction-characteristics}

The preceding section characterizes the weights that satisfy the decrease condition and shows that admissibility alone does not guarantee properness. We now consider how to construct $F$ from a given admissible weight $\kappa$. The construction requires regularity of $\kappa$ and well-defined characteristic curves reaching a prescribed reference line; positive definiteness and properness of $W=F(V,h)$ are then verified separately.

The relation \eqref{eq:weight} can be written as the first-order PDE
\begin{equation}
    F_s(v,s)+\kappa(v,s)F_v(v,s)=0.
    \label{eq:construction-pde}
\end{equation}
For each $(v,s)$, let $\xi(\tau;s,v)$ denote the solution of the characteristic equation
\begin{equation}
    \frac{\mathrm{d}\xi}{\mathrm{d}\tau} = \kappa(\xi,\tau),
    \quad
    \xi(s;s,v)=v.
    \label{eq:construction-characteristic}
\end{equation}
Along the solution, one has
\begin{equation}
\begin{aligned}
    \frac{\mathrm{d}}{\mathrm{d}\tau} F(\xi(\tau;s,v),\tau) &= \kappa(\xi(\tau;s,v),\tau) F_v(\xi(\tau;s,v),\tau) \\
    &\quad +F_s(\xi(\tau;s,v),\tau).
\end{aligned}
\label{eq:F-along-characteristics}
\end{equation}
Hence, $F$ is constant along each characteristic whenever it satisfies \eqref{eq:construction-pde}.

\begin{theorem}[Characteristic Construction]
\label{thm:characteristic-construction}
Let $\mathcal O\subset\R\times(0,\infty)$ be an open set containing $\mathcal R_{V,h}$, and let $\kappa\in C^1(\mathcal O)$ satisfy $\kappa(v,s)\geq0$ for all $(v,s)\in\mathcal R_{V,h}$. Fix a reference value $s_b>0$, and suppose that, for every $(v,s)\in\mathcal O$, the solution $\xi(\tau;s,v)$ of \eqref{eq:construction-characteristic} exists for all $\tau$ between $s$ and $s_b$ and satisfies $(\xi(\tau;s,v),\tau)\in\mathcal O$ throughout that interval. Let $\psi\in C^1(\R)$ satisfy $\psi'(\xi(s_b;s,v))>0$ for every $(v,s)\in\mathcal R_{V,h}$, and define
\begin{equation}
    F(v,s):=\psi\!\left(\xi(s_b;s,v)\right).
    \label{eq:construction-F}
\end{equation}
Then $F\in C^1(\mathcal O)$, satisfies \eqref{eq:construction-pde} on $\mathcal O$, and satisfies \eqref{eq:F-monotonicity} on $\mathcal R_{V,h}$. If, in addition, \eqref{eq:weight-selection} holds and $W=F(V,h)$ is positive definite and proper on $\operatorname{Int}(\C)$, then $W$ is a CLBF.
\end{theorem}

\begin{proof}
Since $\kappa\in C^1(\mathcal O)$, the solution $\xi(s_b;s,v)$ is continuously differentiable with respect to $(s,v)$. Therefore, \eqref{eq:construction-F} implies that $F\in C^1(\mathcal O)$.

For any $\tau$ between $s$ and $s_b$, the flow property of \eqref{eq:construction-characteristic} gives
\begin{equation}
    \xi\!\left(s_b;\tau,\xi(\tau;s,v)\right) = \xi(s_b;s,v).
\end{equation}
Therefore, following \eqref{eq:construction-F}, one has
\begin{equation}
    F\!\left(\xi(\tau;s,v),\tau\right) = \psi\!\left(\xi(s_b;s,v)\right) = F(v,s),
\end{equation}
so $F$ is constant along each characteristic. Using \eqref{eq:F-along-characteristics}, we obtain \eqref{eq:construction-pde}.
It remains to verify the monotonicity of $F$. Differentiating the characteristic equation with respect to the initial value $v$ gives
\begin{equation}
    \frac{\mathrm{d}}{\mathrm{d}\tau}
    \frac{\partial\xi(\tau;s,v)}{\partial v}
    =
    \kappa_v\!\left(\xi(\tau;s,v),\tau\right)
    \frac{\partial\xi(\tau;s,v)}{\partial v}
\end{equation}
with ${\partial\xi(s;s,v)}/{\partial v}=1$. Hence,
\begin{equation}
    \frac{\partial\xi(s_b;s,v)}{\partial v} = \exp\left( \int_s^{s_b} \kappa_v\!\left(\xi(\tau;s,v),\tau\right) \,\mathrm{d}\tau\right) >0.
\end{equation}
For $(v,s)\in\mathcal R_{V,h}$, \eqref{eq:construction-F} and the assumption on $\psi$ give
\begin{equation}
    F_v(v,s) = \psi'\!\left(\xi(s_b;s,v)\right) \frac{\partial\xi(s_b;s,v)}{\partial v} >0.
\end{equation}
Since $\kappa\geq0$ on $\mathcal R_{V,h}$, \eqref{eq:construction-pde} then yields $F_s=-\kappa F_v\leq0$ there. Thus \eqref{eq:F-monotonicity} holds on $\mathcal R_{V,h}$.

Finally, on $\mathcal R_{V,h}$, \eqref{eq:construction-pde} and $F_v>0$ give $-F_s/F_v=\kappa$. Therefore, if \eqref{eq:weight-selection} holds, Lemma~\ref{lem:weight-characterization} gives the CLBF decrease condition. Together with the assumed positive definiteness and properness of $W=F(V,h)$ on $\operatorname{Int}(\C)$, this proves that $W$ is a CLBF.
\end{proof}

Theorem~\ref{thm:characteristic-construction} separates the construction into two parts. The weight $\kappa$ determines the characteristic curves, and hence the ratio $-F_s/F_v$, while $\psi$ assigns values to these curves through the reference line $s=s_b$. Thus, different choices of $\psi$ can generate different functions $F$ associated with the same weight $\kappa$. Admissibility of $\kappa$ guarantees the decrease condition, while positive definiteness and properness of $W$ must still be checked separately. The PDE and characteristic construction in Theorem~\ref{thm:characteristic-construction} are reminiscent of Pomet's construction of time-varying Lyapunov functions for nonholonomic systems~\cite{pomet1992explicit}.

Sontag's formula~\cite{sontag1989universal} can be applied to the constructed CLBF $W=F(V,h)$, yielding an $L_gW$-type feedback that renders $\dot W<0$ away from the origin. Since $W$ has proper sublevel sets contained in the safe-set interior, the resulting closed-loop trajectories remain safe as long as forward solutions are well defined. The construction also provides a natural starting point for extending \textit{inverse-optimal feedback design}~\cite{krstic1998inverse,krstic2024inverse} to safe stabilization through the CLBF $W$. A complete inverse-optimal treatment is left for future work.

\begin{example}[Linear System (continued)]
    We return to the linear system of Example \ref{ex:1}. For the selected weight \eqref{eq:linear-example-kappa}, the characteristic equation is \begin{equation}
    \frac{\mathrm{d}\xi}{\mathrm{d}\tau} = \frac{\xi}{\tau(1+\xi)}.
\end{equation}
For $v>0$, separation of variables gives
\begin{equation}
    \xi(\tau;s,v)e^{\xi(\tau;s,v)} = \frac{\tau}{s}ve^v,
    \label{eq:linear-example-characteristic-invariant}
\end{equation}
and hence
\begin{equation}
    \xi(s_b;s,v)e^{\xi(s_b;s,v)} = \frac{s_b}{s}ve^v.
\end{equation}
Rather than solving explicitly for $\xi(s_b;s,v)$, choose $\psi(r):={re^r}/{s_b}$.
Since $v\geq0$, the characteristic remains nonnegative. Moreover, $\psi'(r)=e^r(1+r)/s_b>0$ for all $r\geq0$. Then \eqref{eq:construction-F} immediately gives
\begin{equation}
    F(v,s)=\frac{ve^v}{s} \quad \text{and} \quad W(x)=\frac{V(x)e^{V(x)}}{h(x)}.
    \label{eq:linear-example-F}
\end{equation}
Since $h>0$ on $\operatorname{Int}(\C)$ and $V$ is positive definite, $W$ is positive definite. It remains only to verify properness. From \eqref{eq:linear-example-R}, $v\geq(1-s)^2/4$, so $W\to\infty$ as $h=s\to0^+$. Moreover, $h(x)\leq1+2\sqrt{V(x)}$, and therefore
\begin{equation}
    W(x) \geq \frac{V(x)e^{V(x)}}{1+2\sqrt{V(x)}} \to\infty \quad\text{as}~|x|\to\infty.
\end{equation}
Thus $W$ is proper. Since $\kappa(v,s)\in\mathcal A(v,s)$, Theorem~\ref{thm:characteristic-construction} shows that $W$ is a CLBF.
\hfill $\blacktriangle$
\end{example}

\subsection{Closed-Form Constructions}
\label{sec:closed-form-constructions}

Theorem~\ref{thm:characteristic-construction} applies to general admissible weights. For weights of the separable form
\begin{equation}
    \kappa(v,s)=\gamma(v)\mu(s),
    \label{eq:separable-weight}
\end{equation}
the characteristic equation can be integrated explicitly. For $v>0$, assume $\gamma(v)>0$ and define
\begin{equation}
    G(v):=\int_{v_\star}^{v}\frac{\mathrm{d}r}{\gamma(r)},
    \qquad v_\star>0.
    \label{eq:G-definition}
\end{equation}
Along a characteristic,
\begin{equation}
    \frac{\mathrm{d}}{\mathrm{d}\tau} G(\xi(\tau;s,v)) = \mu(\tau),
\end{equation}
and therefore
\begin{equation}
    G(\xi(s_b;s,v))=G(v)+\int_s^{s_b}\mu(\tau)\,\mathrm{d}\tau.
    \label{eq:separable-characteristic}
\end{equation}
Choosing $\psi(r)=r$ in Theorem~\ref{thm:characteristic-construction} gives the closed-form construction
\begin{equation}
    F(v,s) = G^{-1}\!\left( G(v)+\int_s^{s_b}\mu(\tau)\,\mathrm{d}\tau \right),
    \label{eq:separable-F}
\end{equation}
whenever the right-hand side is well defined.

An important special case is $\gamma(v)=v$. Then
\begin{equation}
    F(v,s) = v\exp\left(\int_s^{s_b}\mu(\tau)\,\mathrm{d}\tau\right),
    \label{eq:multiplicative-F}
\end{equation}
with $F(0,s)=0$, and hence
\begin{equation}
    \boxed{W(x) = V(x)\exp\left(\int_{h(x)}^{s_b}\mu(\tau)\,\mathrm{d}\tau\right).}
    \label{eq:multiplicative-W}
\end{equation}

\begin{corollary}[Multiplicative Construction]
\label{cor:multiplicative-construction}
Let $V$ be a global CLF for \eqref{eq:system}, and let $h$ be a CBF with respect to $\C$. Let $s_b>0$ and let $\mu\in C^1((0,\infty);\R_{\geq0})$. Suppose that
\begin{equation}
    v\mu(s)\in\mathcal A(v,s), \quad \forall (v,s)\in\mathcal R_{V,h},
    \label{eq:multiplicative-admissibility}
\end{equation}
and that
\begin{equation}
    \inf_{s\in h(\operatorname{Int}(\C))} \int_s^{s_b}\mu(\tau)\,\mathrm{d}\tau >-\infty.
    \label{eq:multiplicative-radial}
\end{equation}
Assume further that
\begin{equation}
    \int_s^{s_b}\mu(\tau)\,\mathrm{d}\tau \to\infty \quad \text{as}~s\to0^+.
    \label{eq:multiplicative-boundary}
\end{equation}
Then \eqref{eq:multiplicative-W} is a CLBF on $\operatorname{Int}(\C)$.
\end{corollary}

\begin{proof}
Let $F$ be defined in \eqref{eq:multiplicative-F}. Then
\begin{equation}
    F_v(v,s) = \exp\left( \int_s^{s_b}\mu(\tau)\,\mathrm{d}\tau \right)>0.
\end{equation}
For $(v,s)\in\mathcal R_{V,h}$, one has $v\geq0$, and hence
\begin{equation}
    F_s(v,s) = -v\mu(s) \exp\left( \int_s^{s_b}\mu(\tau)\,\mathrm{d}\tau \right) \leq0.
\end{equation}
Moreover, $-{F_s(v,s)}/{F_v(v,s)} = v\mu(s)$. Since the exponential factor is positive and $V$ is positive definite, $W=F(V,h)$ is positive definite.

It remains to verify properness. By \eqref{eq:multiplicative-radial}, there exists $M\in\R$ such that
\begin{equation}
    \int_{h(x)}^{s_b}\mu(\tau)\,\mathrm{d}\tau \geq M, \quad \forall x\in\operatorname{Int}(\C).
\end{equation}
Therefore, $W(x)\geq e^M V(x)$. Hence, for any $r\geq0$, $W(x)\leq r \implies V(x)\leq e^{-M}r$. Since $V$ is globally proper, the sublevel set $\{x:V(x)\leq e^{-M}r\}$ is compact.

Now suppose that $W(x_j)\leq r$ and $x_j\to x^\ast\in\partial\C$. Then $h(x_j)\to0^+$. Since $0\in\operatorname{Int}(\C)$, one has $x^\ast\neq0$, and hence $V(x_j)\to V(x^\ast)>0$. By \eqref{eq:multiplicative-boundary},
\begin{equation}
    \exp\left( \int_{h(x_j)}^{s_b}\mu(\tau)\,\mathrm{d}\tau \right) \to\infty,
\end{equation}
which implies $W(x_j)\to\infty$, contradicting $W(x_j)\leq r$. Thus no sublevel set of $W$ can approach $\partial\C$. Together with the compactness of the corresponding sublevel set of $V$, this shows that every sublevel set of $W$ is compact in $\operatorname{Int}(\C)$. Hence $W$ is proper.

Finally, by \eqref{eq:multiplicative-admissibility} and $-{F_s(v,s)}/{F_v(v,s)}=v\mu(s)$, Lemma~\ref{lem:weight-characterization} gives the CLBF decrease condition. Therefore, $W$ is a CLBF on $\operatorname{Int}(\C)$.
\end{proof}

Corollary~\ref{cor:multiplicative-construction} allows a general nonnegative function $\mu$. The special choice
\begin{equation}
    \mu(s)=\frac{c}{\alpha_h(s)}
    \label{eq:normalized-mu}
\end{equation}
with $c>0$ is particularly useful when the CLF and CBF are \textit{compatible} on $\operatorname{Int}(\C)\setminus\{0\}$, in the sense that for every $x\in\operatorname{Int}(\C)\setminus\{0\}$ there exists a common $u\in\R^m$ such that
\begin{subequations}
\begin{align}
    \Lie{f}{V}(x)+\Lie{g}{V}(x)u
    &<-\alpha(|x|),
    \label{eq:compatibility-clf}\\
    \Lie{f}{h}(x)+\Lie{g}{h}(x)u
    &>-\alpha_h(h(x)).
    \label{eq:compatibility-cbf}
\end{align}
\end{subequations}

\begin{proposition}[Construction under Compatibility]
\label{prop:compatibility-construction}
Let $V$ be a global CLF for \eqref{eq:system}, satisfying \eqref{eq:defCLF} with a continuous positive definite function $\alpha$, and let $h$ be a CBF with respect to $\C$, associated with $\alpha_h\in C^1(\R)\cap\mathcal K^e$. Assume that
\begin{equation}
    \bar h
    :=
    \sup_{x\in\operatorname{Int}(\C)} h(x)
    <\infty.
    \label{eq:h-bounded}
\end{equation}
For any $s_\star>0$, define
\begin{equation}
    \Psi(s)
    :=
    \int_s^{s_\star}
    \frac{\mathrm{d}\tau}{\alpha_h(\tau)},
    \qquad s>0,
    \label{eq:Psi-definition}
\end{equation}
and suppose that
\begin{equation}
    \Psi(s)\to\infty
    \quad
    \text{as}~s\to0^+.
    \label{eq:Psi-divergence}
\end{equation}
Suppose further that the CLF and CBF are compatible on $\operatorname{Int}(\C)\setminus\{0\}$ in the sense of \eqref{eq:compatibility-clf}--\eqref{eq:compatibility-cbf}. If the same constant $c>0$ in \eqref{eq:normalized-mu} satisfies
\begin{equation}
    cV(x)<\alpha(|x|),
    \quad
    \forall x\in\operatorname{Int}(\C)\setminus\{0\},
    \label{eq:cV-alpha}
\end{equation}
then
\begin{equation}
    W(x)
    :=
    V(x)\exp\left(
        c\bigl[\Psi(h(x))-\Psi(h(0))\bigr]
    \right)
    \label{eq:compatibility-W}
\end{equation}
is a CLBF on $\operatorname{Int}(\C)$.
\end{proposition}

\begin{proof}
Apply Corollary~\ref{cor:multiplicative-construction} with $s_b=h(0)$ and $\mu$ given by \eqref{eq:normalized-mu}. Since $\alpha_h(s)>0$ for $s>0$, one has $\mu(s)\geq0$. By \eqref{eq:Psi-definition},
\begin{equation}
    \int_s^{h(0)}\mu(\tau)\,\mathrm{d}\tau
    =
    c\bigl[\Psi(s)-\Psi(h(0))\bigr].
    \label{eq:normalized-integral}
\end{equation}
Condition \eqref{eq:Psi-divergence} gives the required boundary divergence. Moreover, $\Psi$ is decreasing on $(0,\infty)$, and $h(x)\leq\bar h$ implies
\begin{equation}
    \Psi(h(x))-\Psi(h(0))
    \geq
    \Psi(\bar h)-\Psi(h(0)).
\end{equation}
Hence the integral in \eqref{eq:normalized-integral} is bounded below on $h(\operatorname{Int}(\C))$, so the two properness conditions of Corollary~\ref{cor:multiplicative-construction} hold.

It remains to verify admissibility. Fix $(v,s)\in\mathcal R_{V,h}$ and $x\in\Sigma(v,s)\setminus\{0\}$, and suppose that
\begin{equation}
    \Lie{g}{V}(x) - \frac{cv}{\alpha_h(s)} \Lie{g}{h}(x) = 0.
    \label{eq:prop2-zero-control-term}
\end{equation}
Since $v=V(x)$ and $s=h(x)$, let $\lambda:=cV(x)/\alpha_h(h(x))$. By compatibility, there exists a common control $u$ satisfying \eqref{eq:compatibility-clf}--\eqref{eq:compatibility-cbf}. Using \eqref{eq:prop2-zero-control-term},
\begin{align}
    \Lie{f}{V}(x)-\lambda\Lie{f}{h}(x)
    &= \bigl(\Lie{f}{V}(x)+\Lie{g}{V}(x)u\bigr) \nonumber\\
    &\qquad -\lambda\bigl(\Lie{f}{h}(x)+\Lie{g}{h}(x)u\bigr)
    \nonumber\\
    &< -\alpha(|x|) + \lambda\alpha_h(h(x)) \nonumber\\
    &= -\alpha(|x|)+cV(x) <0,
\end{align}
where the last inequality follows from \eqref{eq:cV-alpha}. Therefore, $cv/\alpha_h(s)\in\mathcal A(v,s)$ for every $(v,s)\in\mathcal R_{V,h}$. Corollary~\ref{cor:multiplicative-construction} then yields \eqref{eq:compatibility-W}.
\end{proof}

If $\alpha_h(s)=bs$ with $b>0$, then \eqref{eq:compatibility-W} reduces to
\begin{equation}
    \boxed{W(x) = V(x) \left( \frac{h(0)}{h(x)} \right)^{c/b}.}
    \label{eq:power-composition}
\end{equation}
Thus, under the assumptions of Proposition~\ref{prop:compatibility-construction}, the choice \eqref{eq:normalized-mu} gives the \textit{power-type CLBF} directly from Corollary~\ref{cor:multiplicative-construction}. 

In summary, Proposition~\ref{prop:compatibility-construction} provides the following construction procedure. Given $V$, $h$, $\alpha$, and $\alpha_h$, verify \eqref{eq:h-bounded} and \eqref{eq:Psi-divergence}, check CLF--CBF compatibility, and choose $c>0$ satisfying \eqref{eq:cV-alpha}. Then set $\mu(s)=c/\alpha_h(s)$ and construct $W$ from \eqref{eq:compatibility-W}; when $\alpha_h(s)=bs$, this reduces to \eqref{eq:power-composition}.

\section{Examples of CLBF Construction}\label{sec:examples}

This section gives two nonlinear examples. The first applies Proposition~\ref{prop:compatibility-construction} to construct a multiplicative CLBF from a compatible CLF--CBF pair. The second uses the characteristic construction of Theorem~\ref{thm:characteristic-construction} and illustrates the additional freedom available when multiple weights are admissible. In both examples, the conditions required by the corresponding construction are verified analytically before the feedback laws are introduced; numerical integration is used only to illustrate the resulting closed-loop trajectories. A feedback based on the original CLF stabilizes the origin but may violate the safety constraint, whereas the feedback based on the constructed CLBF achieves \textit{safe stabilization}.

\begin{example}[Multiplicative CLBF under Compatibility]
Consider the nonlinear system
\begin{align}
    \dot x_1=x_2-x_2^2u, \qquad \dot x_2=u.
    \label{eq:example1-system}
\end{align}
Define $z:=x_1+x_2+x_2^3/3$ and
\begin{equation*}
    V(x):=\frac{1}{2}\left(z^2+x_2^2\right),
    \quad
    h(x):=1-z^2.
\end{equation*}
The map $x\mapsto(z,x_2)$ is a global diffeomorphism, and in $z$ coordinate $\dot z=x_2+u$. The safe set is $\C=\{x:|z|\leq1\}$. The corresponding Lie derivatives are $\Lie{f}{V}=zx_2$, $\Lie{g}{V}=z+x_2$, $\Lie{f}{h}=-2zx_2$, and $\Lie{g}{h}=-2z$.

We verify the conditions of Proposition~\ref{prop:compatibility-construction} for the transformed system with state $(z,x_2)$; accordingly, the state norm in the CLF decay condition and in \eqref{eq:cV-alpha} is $|(z,x_2)|$. Function expressions and initial conditions written in $x$ retain the original coordinates $x=(x_1,x_2)$. Fix $c\in(0,1)$, choose $a\in(c/2,1/2)$, and let $\alpha(s):=as^2$. If $\Lie{g}{V}=0$, then $x_2=-z$, and hence $\Lie{f}{V}+\alpha(|(z,x_2)|)=-(1-2a)z^2<0$ away from the origin. Together with the positive definiteness and properness of $V$, this shows that $V$ is a global CLF. Moreover, $\Lie{g}{h}=0$ implies $z=0$, for which $\Lie{f}{h}=0$ and $h=1$. Thus, $h$ is a CBF associated with $\alpha_h(s)=s$.

It remains to verify compatibility. Let $r:=x_2+u$. The two inequalities \eqref{eq:compatibility-clf}--\eqref{eq:compatibility-cbf} are equivalent to
\begin{subequations}
\begin{align}
    (z+x_2)r &< (1-a)x_2^2-az^2,
    \label{eq:example1-compatibility-clf}\\
    2zr &< 1-z^2.
    \label{eq:example1-compatibility-cbf}
\end{align}
\end{subequations}
If $z(z+x_2)\leq0$, then $x_2^2\geq z^2$ whenever $z\neq0$, and since $a<1/2$ the right-hand side of \eqref{eq:example1-compatibility-clf} is positive away from the origin. Thus $r=0$ satisfies both inequalities. If $z(z+x_2)>0$, then $z$ and $z+x_2$ have the same sign, so choosing $r$ with the opposite sign and sufficiently large magnitude makes the left-hand sides of both \eqref{eq:example1-compatibility-clf} and \eqref{eq:example1-compatibility-cbf} negative with arbitrarily large magnitude. Hence the CLF and CBF are compatible on $\operatorname{Int}(\C)\setminus\{0\}$.

Since $V=(z^2+x_2^2)/2$ and $a>c/2$, one obtains $cV(x)<\alpha(|(z,x_2)|)$ for $x\neq0$, so \eqref{eq:cV-alpha} holds. Furthermore, $\bar h=1$, and with $s_\star=1$,
\begin{equation}
    \Psi(s)=\int_s^1\frac{\mathrm{d}\tau}{\tau}=\ln\frac{1}{s}\to\infty \quad    \text{as}~s\to0^+.
    \label{eq:example1-Psi}
\end{equation}
Proposition~\ref{prop:compatibility-construction} therefore gives
\begin{equation}
    \boxed{W(x) = \frac{z^2+x_2^2}{2(1-z^2)^c}.}
    \label{eq:example1-W}
\end{equation}

For comparison, consider the CLF damping feedback
\begin{equation}
    u_V=-k\Lie{g}{V}(x)=-k(z+x_2).
    \label{eq:example1-uV}
\end{equation}
With $k>1/4$, it yields $\dot V=zx_2-k(z+x_2)^2<0$ away from the origin, but need not preserve safety with respect to $\C$. For example, with $k=0.3$, $c=0.5$, and $x(0)=(-1,1.2)$, the resulting trajectory crosses the boundary $z=1$, as shown in Fig.~\ref{fig:example-1}.

For the CLBF \eqref{eq:example1-W}, define
\begin{equation}
    p(x):=z\left(1+\frac{2cV(x)}{h(x)}\right).
    \label{eq:example1-p}
\end{equation}
Then $\Lie{g}{W}=h^{-c}(p+x_2)$ and $\Lie{f}{W}=h^{-c}px_2$. Hence,
\begin{equation}
    u_W=-kh^c\Lie{g}{W}(x)=-k(p+x_2)
    \label{eq:example1-uW}
\end{equation}
gives $\dot W=h^{-c}\left[ px_2-k(p+x_2)^2 \right]<0$ for every nonzero $x\in\operatorname{Int}(\C)$ whenever $k>1/4$. Since $W$ is proper on $\operatorname{Int}(\C)$, the feedback \eqref{eq:example1-uW} asymptotically stabilizes the origin while preserving safety.
\hfill $\blacktriangle$
\end{example}

\begin{figure}[t]
    \centering
    \includegraphics[width=0.7\linewidth]{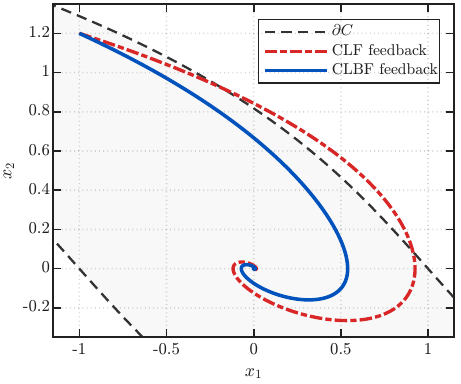}
    \caption{Closed-loop trajectories for Example~3. The CLF feedback
    crosses the safety boundary $z=1$, whereas the CLBF feedback remains
    in $\operatorname{Int}(\C)$ and converges to the origin.}
    \label{fig:example-1}
\end{figure}

\begin{example}[Construction via Characteristics]
    Consider the nonlinear system
\begin{equation}
    \dot x_1=x_2^3+u, \qquad \dot x_2=u.
    \label{eq:example2-system}
\end{equation}
Let
\begin{equation*}
    V(x):=\frac{1}{2}x_1^2+\frac{1}{4}x_2^4,
    \qquad
    h(x):=1-x_1^4.
\end{equation*}
The safe set is $\C=\{x:|x_1|\leq1\}$. The corresponding Lie derivatives are $\Lie{f}{V}=x_1x_2^3$, $\Lie{g}{V}=x_1+x_2^3$, $\Lie{f}{h}=-4x_1^3x_2^3$, and $\Lie{g}{h}=-4x_1^3$. If $\Lie{g}{V}=0$, then $x_2^3=-x_1$ and hence $\Lie{f}{V}=-x_1^2<0$ away from the origin. Thus, $V$ is a global CLF. Condition \eqref{eq:defCLF} holds with $\alpha(s):=s^6/(2(1+s^4))$. Moreover, $\Lie{g}{h}=0$ implies $x_1=0$, for which $\Lie{f}{h}=0$ and $h=1$. Therefore, $h$ is a CBF associated with $\alpha_h(s)=s$.

We first determine the admissible weights. For an arbitrary $k\geq0$, define $p_k:=x_1(1+4kx_1^2)$. Then
\begin{subequations}
\begin{align}
    \Lie{g}{V}-k\Lie{g}{h}&=x_2^3+p_k,
    \label{eq:example2-mixed-g}\\
    \Lie{f}{V}-k\Lie{f}{h}&=p_kx_2^3.
    \label{eq:example2-mixed-f}
\end{align}
\end{subequations}
If \eqref{eq:example2-mixed-g} vanishes, then $x_2^3=-p_k$, so \eqref{eq:example2-mixed-f} becomes $-p_k^2$. Since $k\geq0$, $p_k=0$ implies $x_1=0$, which together with $x_2^3=-p_k$ implies $x=0$. Hence
\begin{equation}
    \mathcal A(v,s)=\R_{\geq0},\quad\forall (v,s)\in\mathcal R_{V,h}.
    \label{eq:example2-all-admissible}
\end{equation}
This leaves considerable freedom in choosing $\kappa$.

A simple choice is $\kappa(v,s)=cv/s$ for any $c>0$. Taking $\mu(s)=c/s$ and $s_b=1$ in Corollary~\ref{cor:multiplicative-construction} yields the family
\begin{equation}
    \boxed{W_c(x)=\frac{V(x)}{h(x)^c}.}
    \label{eq:example2-power}
\end{equation}
Since $0<h(x)\leq1$ on $\operatorname{Int}(\C)$, so $\int_{h(x)}^1 c/\tau\,\mathrm{d}\tau\geq0$ and this integral diverges as $h(x)\to0^+$. Hence, following Corollary~\ref{cor:multiplicative-construction}, \eqref{eq:example2-power} is a CLBF. This family also shows that the sufficient decay condition \eqref{eq:cV-alpha} in Proposition~\ref{prop:compatibility-construction} is not necessary. Along $x_1=0$, one gets ${\alpha(|x|)}/{V(x)}={2x_2^2}/{(1+x_2^4)}\to 0$ as $x_2\to 0$. Thus, no fixed $c>0$ satisfies $cV(x)<\alpha(|x|)$ throughout $\operatorname{Int}(\C)\setminus\{0\}$, although \eqref{eq:example2-power} is a CLBF for every $c>0$.

The freedom in the admissible weight also gives nonmultiplicative constructions. Let $\beta\in C^1((0,\infty);\R_{\geq0})$ satisfy
\begin{equation}
    \int_s^1\beta(\tau)\,\mathrm{d}\tau\to\infty \quad \text{as} \quad s\to0^+.
    \label{eq:example2-beta-condition}
\end{equation}
Choosing $\kappa(v,s):=\beta(s)$, $s_b=1$, and $\psi(r):=r$ in Theorem~\ref{thm:characteristic-construction} gives the CLBF family
\begin{equation}
    W_\beta(x)=V(x)+\int_{h(x)}^1\beta(\tau)\,\mathrm{d}\tau,
    \label{eq:example2-Wbeta}
\end{equation}
which is positive definite and proper on $\operatorname{Int}(\C)$. Indeed, the integral term diverges as $h(x)\to0^+$, while $V(x)\to\infty$ whenever $|x_2|\to\infty$ with $|x_1|<1$. For example, the family $\beta_q(s):=s^{-q}$, $q\geq1$, yields CLBF family
\begin{equation}
    \boxed{W_q(x)=
    \begin{cases}
        V(x)-\ln h(x), & q=1,\\[1mm]
        V(x)+\dfrac{h(x)^{1-q}-1}{q-1}, & q>1.
    \end{cases}}
    \label{eq:example2-Wq}
\end{equation}

The weight may also depend on both $v$ and $s$. For instance, consider $\kappa(v,s):={\sqrt{1+v}}/{s}$. This weight is also admissible. Taking $s_b=1$ and $\psi(r)=r$, solving \eqref{eq:construction-characteristic} gives
\begin{equation}
    \boxed{W(x)=\left(\sqrt{1+V(x)}+\frac{1}{2}\ln\frac{1}{h(x)}\right)^2-1.}
    \label{eq:example2-W-coupled}
\end{equation}
This provides a nonmultiplicative construction whose weight depends on both $v$ and $s$.

For any of the preceding constructions, define
\begin{equation}
    p(x):=x_1\left(1+4\kappa(V,h)x_1^2\right).
    \label{eq:example2-p}
\end{equation}
The identity $F_s+\kappa F_v=0$ gives $\Lie{g}{W}=F_v(x_2^3+p)$ and $\Lie{f}{W}=F_vpx_2^3$. Hence, the CLBF feedback
\begin{equation}
    u_W=-\frac{l}{F_v(V,h)}\Lie{g}{W}=-l(p+x_2^3)
    \label{eq:example2-uW}
\end{equation}
with $l>1/4$ satisfies $\dot W=F_v\left[px_2^3-l(p+x_2^3)^2\right]<0$ for every $x\in\operatorname{Int}(\C)\setminus\{0\}$.

In contrast, the CLF feedback $u_V=-l\Lie{g}{V}=-l(x_1+x_2^3)$ with $l>1/4$ decreases $V$ but need not preserve safety with respect to $\C$. For the numerical comparison, we use \eqref{eq:example2-Wq} with $q=1$, namely $W=V-\ln h$. Both controllers use $l=0.3$ and $x(0)=(0.9,1.2)\in\operatorname{Int}(\C)$. As shown in Fig.~\ref{fig:example-2}, the CLF feedback drives the trajectory across the boundary $x_1=1$, whereas the CLBF feedback keeps the trajectory in $\operatorname{Int}(\C)$ and steers it toward the origin.
\hfill $\blacktriangle$
\end{example}

\begin{figure}[t]
    \centering
    \includegraphics[width=0.7\linewidth]{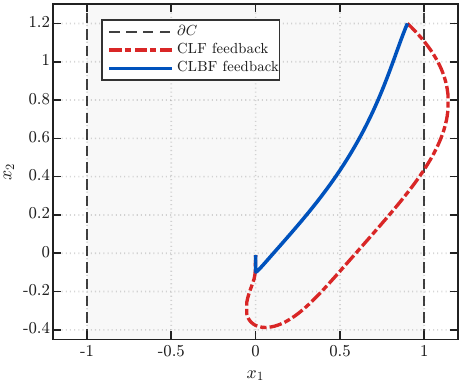}
    \caption{Closed-loop trajectories for Example~4. The CLF feedback
    drives the trajectory across the safety boundary $x_1=1$, whereas
    the CLBF feedback preserves $\operatorname{Int}(\C)$ and steers the
    state toward the origin.}
    \label{fig:example-2}
\end{figure}

\section{Conclusion}\label{sec:conclusion}

This paper studied the construction of CLBFs of the form $W=F(V,h)$ from a given CLF $V$ and CBF $h$. For functions with $F_v>0$ and $F_s\leq0$, we showed that the CLBF decrease condition is completely characterized by a nonnegative scalar weight and its admissibility on joint level sets of $(V,h)$. Given an admissible weight, we constructed $F$ through a first-order PDE and its characteristic curves, yielding general constructions as well as explicit multiplicative and power-type families under suitable conditions. We also showed that admissibility alone does not guarantee properness and established an integrability obstruction for weights that cannot generate the required growth near the safety boundary. Examples illustrated these constructions and demonstrated safe stabilization in cases where feedback based on the original CLF violates the safety constraint.

\bibliographystyle{IEEEtran}
\bibliography{mybibfile}

\end{document}